\documentclass[12pt]{amsart}
\usepackage{tikz,blkarray,multirow}
\usepackage{graphicx,euscript,verbatim,constants}
\usepackage{amssymb,url,booktabs}
\usepackage[top=3cm,bottom=2cm,left=2.5cm,right=2.5cm]{geometry}

\usepackage{listings}
\usepackage{euscript}
\usepackage{subfigure}
\usepackage{graphicx}
\usepackage{float}
\usepackage{subfigure}
\usepackage{amsmath}
\usepackage{amsthm}
\usepackage{amssymb}
\usepackage{multirow,url}
\usepackage{pict2e,slashbox,parallel}
\usepackage{booktabs,chngpage}
\usepackage{commath}
\usepackage{textcomp}
\usepackage{float}
\usepackage{nicefrac}
\usepackage[section]{placeins} 
\usepackage{longtable}
\usepackage{microtype}
\usepackage{enumerate}
\theoremstyle{plain}

\newtheorem{Thm}{Theorem}

\newtheorem{Claim}[Thm]{Claim}
\newtheorem{lemma}[Thm]{Lemma}

\newcommand{\cond}{\,\lv\,}

\newcommand{\wt}{\widetilde}

\newcommand{\E}{\mathbb{E}}

\newcommand{\gs}{\sigma}

\newcommand{\gl}{\lambda}

\newcommand{\gep}{\epsilon}

\newcommand{\nc}{\newcommand}

\newcommand{\eu}{\EuScript}
\newcommand{\indic}{\boldsymbol{1}}

\newcommand{\on}{\operatorname}
\nc{\G}{\eu{G}}
\nc{\lip}{\on{Lip}}
\nc{\izf}{\int_0^\infty}
\nc{\imf}{\int_{-\infty}^\infty}
\nc{\tand}{\text{ and }}
\nc{\tst}{\text{ s.t. }}
\nc{\fM}{\mathfrak{M}}
\nc{\fP}{\mathfrak{P}}

\nc{\bE}{\mathbb{E}}

\newcommand{\lv}{\bigl|}

\newcommand{\rec}{\frac{1}}

\newcommand{\R}{\mathbb{R}}

\nc{\N}{\mathbb{N}}
\nc{\mL}{\eu{L}}
\nc{\mA}{\eu{A}}
\nc{\M}{\eu{M}}
\nc{\B}{\eu{B}}

\nc{\vx}{\vec{x}}
\nc{\vy}{\vec{y}}
\nc{\DF}{\eu{F}}
\nc{\tX}{\wt{X}}
\nc{\mE}{\mathbb{E}}
\nc{\brM}{\bar{\mM}}
\nc{\tih}{\tilde{h}}
\nc{\lep}{\frac{\gl}{\gep}}
\nc{\tp}{\tau_{\partial}}
\nc{\sM}{\mM^{*}}
\nc{\ns}{\nu^*}
\renewcommand{\P}{\mathbb{P}}
\nc{\ulm}{\underline{\gl}}
\nc{\Lip}{\on{Lip}}
\nc{\Was}{\on{Was}}
\nc{\salg}{$\gs$-algebra }
\nc{\salgns}{$\gs$-algebra}
\nc{\salgs}{$\gs$-algebras }
\nc{\pt}{(\Omega,\eu{F},P)}
\nc{\ptt}{Let $\pt$ be a probability triple }
\nc{\bfX}{\mathbf{X}}
\nc{\rnfor}{}
\nc{\rnfig}{}
\nc{\pdf}{\rec{\sqrt{2\pi}} e^{-z^{2}/2}}
\nc{\sol}{\smallskip \noindent{\bf Solution}: }
\nc{\ttR}{{\tt R}}
\nc{\var}{\on{Var}}
\nc{\ul}{\underline\lambda}
\nc{\nce}{\operatorname{e}}

\nc{\cenplus}{\raisebox{1.5pt}{\text{\scriptsize +}}}

\nc{\sa}{$\sigma$-algebra}

\nc{\e}{\mathrm{e}}

\nc{\bX}{\mathbf{X}}
\renewcommand{\mE}{\mathcal{E}}
\nc{\tmu}{\tilde{\mu}}
\nc{\tx}{\tilde{\xi}}
\nc{\ttau}{\tau}
\nc{\tsig}{\tilde{\sigma}}
\nc{\hA}{\hat{A}}
\nc{\hR}{\hat{R}}
\nc{\ha}{\hat{a}}
\nc{\hD}{\hat{D}^{*}}
\nc{\bA}{\mathbf{A}}

\nc{\constc}{c}
\nc{\constcp}{c'}
\nc{\J}{\eu{J}}
\nc{\tB}{\wt{B}}
\nc{\tf}{\tilde{f}}
\nc{\tD}{\wt{\Delta}}
\nc{\U}{\mathcal{U}}
\nc{\V}{\mathcal{V}}
\nc{\tU}{\check{U}}
\nc{\tV}{\check{V}}
\nc{\tN}{\wt{N}}
\nc{\tY}{\wt{Y}}

\begin{document}
\title[Stochastic growth rates with rare migration]{Stochastic growth rates for populations in random environments with rare migration}
\author{David Steinsaltz} 
\author{Shripad Tuljapurkar}
\address{David Steinsaltz\\Department of Statistics\\University of Oxford\\24--29 St Giles\\Oxford OX1 2HB\\United Kingdom}
 \address{Shripad Tuljapurkar\\454 Herrin Labs\\Department of Biology\\Stanford University\\Stanford CA 94305-5020\\USA}

\begin{abstract}
The growth of a population divided among spatial sites, with migration between the sites, is sometimes modelled by a product of random matrices, with each diagonal elements representing the growth rate in a given time period, and off-diagonal elements the migration rate.  The randomness of the matrices then represents stochasticity of environmental
conditions. We consider the case where the off-diagonal elements are small, representing a situation where migration has been introduced into an otherwise sessile meta-population. We examine the asymptotic behaviour of the long-term growth rate. When there is a single site with the highest growth rate, under the assumption of Gaussian log growth rates
at the individual sites (or having Gaussian-like tails)
we show that the behavior near zero is like a power of $\epsilon$, and derive upper and lower bounds for the power in terms of the difference in the growth rates and the distance between the sites. In particular, when the difference in mean log growth rate between two sites is sufficiently small, or the variance of the difference between the sites
sufficiently large, migration will always be favored by natural selection, in the sense that introducing a small amount of migration will increase the growth rate of the population relative to the zero-migration case.
\end{abstract}
\maketitle

\section{Introduction}
\subsection{Biological motivation}
If a population is divided among spatial sites with distinct fixed growth rates,
with no migration between sites, the numbers in the best site will become overwhelmingly larger than
those at the other sites, and the overall population growth rate will be determined by 
the rate prevailing at the best site. Introducing migration between sites, as Karlin showed \cite{karlin82},
will always reduce the long-run growth rate of the total population.

Karlin's theorem assumes deterministic growth. den Boer \cite{boer1968spreading} argued that migration may increase long-run growth when there is independent or weakly correlated stochastic variation in growth among sites.
But Cohen \cite{cohen1966} and Cohen and Levin \cite{cohen1991} used analysis and simulations to show that long-run growth of a population could increase as a result of a life cycle delay when there are some kinds of random variation in time, or by migration when there are some kinds of random variation across space. 
These kinds of stochastic variation have been formulated as random matrix models  whose Lyapunov exponent  is the long-run growth rate of the population, as discussed by \cite{tuljapurkar2000ets, Wiener1994}. In this general setting, we would like to know whether the long-run growth rate increases when there is mixing in space and/or time \cite{tuljapurkar2000ets} --- biologically, when should migration and/or delay be favoured to
evolve? A general and precise answer has been difficult because previous work \cite{Wiener1994} shows that the long-run growth rate can be singular (e.g., non-differentiable) in the limit of no mixing. A similar singularity arises in random-matrix models used in models of disordered matter \cite{derrida1983singular}. 

In the companion paper \cite{diapause2018} we consider a simple model of migration among multiple sites, where two or more
sites have the same optimal average log growth rate. We show there that a small increase from zero migration
to migration at a small rate $\epsilon$ is associated with an increase on the order of $1/\log\epsilon^{-1}$,
a change that overwhelms any cost of migration that is on the order of $\epsilon$ itself. 
As discussed there, while such a specification strains credulity when our life-history story is of individuals
migrating among independently varying sites or patches, it arises naturally when we turn from
geographic to demographic structure, reinterpreting ``sites'' as age classes. Rare migration becomes, in this
framework, rare diapause, a rare random delay in an otherwise deterministic life history.

In this paper we consider the more generic situation for migration, where there is a single optimal
site, where the mean log growth rate is highest, and then one or more alternative sites where growth
is slower on average. We show that, under some plausible conditions, the increase in population
growth rate with migration at rate $\epsilon$ (from $\epsilon=0$) is approximately 
proportional to a power $\epsilon^p$. We can bound $p$, yielding conditions under which $p<1$, making
small deviations from zero migration advantageous in spite of migration costs on the order of $\epsilon$.
Our results complement the analysis in \cite{ERSS12} 
of optimal migration rates for populations divided
among sites with varying stochastic growth rates.
There interacting diffusions are used to characterize the migration rate that maximizes the long-run stochastic growth rate. 

\subsection{Notation and basic assumptions} \label{sec:assumptions}
Suppose $D_{1},D_{2},\dots$ is an i.i.d.\  sequence of $d\times d$ diagonal matrices,
representing population growth rates at $d$ separate sites in a succession of times.
We write $\xi_{t}^{(0)},\dots,\xi_{t}^{(d-1)}$ for the diagonal elements of
$D_{t}$, and assume $X_{t}^{(i)}:=\log\xi_{t}^{(i)}$ all have finite 
mean $\mu_{i}$ and  finite variance $\tau_i$. We order them so that $\mu_{0}$ 
is the largest. We also write $\tmu^{(i)}:=\mu_{0}-\mu_{i}$. We assume that $\tmu^{(i)}>0$ for all $i=1,\dots,d-1$.

We will be assuming throughout that $X_{t}^{(0)}$ is Gaussian with mean $\mu_0$ and variance $\tau^{(0)}<\infty$, and that for $j\ge 1$
$$
\tX_{t}^{(j)}:=X_{t}^{(j)} - X_{t}^{(0)},
$$ 
is Gaussian with mean $-\tmu^{(j)}$ and variance $\tau^{(j)}<\infty$.
This assumption is made to simplify the notation in the proofs. It would suffice to assume these
variables to be sub-Gaussian, in which case different versions of the sub-Gaussian variance factor
would appear in the upper and lower bounds. The notation for sub-Gaussian random variables,
and the appropriate modification of the main result, are outlined briefly in section \ref{sec:Orlicz}.

We write $\bX$ for the complete collection of all $X^{(j)}_{t}$ for $j=0,1,\dots,d-1$,
$0\le t<\infty$.

We define the {\em migration graph} $\M$ to be a simple and irreducible directed graph 
whose vertices are the sites $\{0,\dots,d-1\}$, representing the transitions that have nonzero
probability.
We let $A_{t}$ be an i.i.d.\ sequence of nonnegative
$d\times d$ matrices with zeros on the diagonal, representing migration rates in time-interval $t$. 
We follow the convention from the matrix population model literature, that transition
rates from state $i$ to state $j$ are found in matrix entry $(j,i)$. Population distributions
are thus naturally column vectors, and are updated from time $t-1$ to time $t$ by left multiplication.

We assume $A_t(j,i)$ are bounded above almost surely.
We assume that if $i\nrightarrow j$ then $A_t(j,i)$ is identically 0, while for $i\rightarrow j$
$$ 
 \E \bigl[ A_{t}(j,i)\cond D_{t}\bigr]
$$ 
is bounded below almost surely. We assume that
the collection of pairs $(D_{t},A_{t})_{t=0}^\infty$ is jointly independent, but note that 
we do not assume for a given $t$ that $A_{t}$ and $D_{t}$ are independent, only 
that there is a lower bound to how close $A_{t}(j,i)$ can
come to 0 that is independent of all $D_{t}$.

We let  $\Delta_t$ be
a random diagonal matrix with entries $\Delta_t^{0},\dots,\Delta_t^{d-1}$. (Generally we will be thinking of $\Delta$ as the growth or survival
penalty for migration, so that the entries will be negative, but this
is not essential.) We assume the penalty acts multiplicatively on growth ---
this seems reasonable from a modeling perspective, and avoids the 
problem of negative matrix entries --- and is proportional to $\epsilon$. We assume that these
penalties are almost surely bounded, with $\|\Delta\|:= \max_{i,j} \on{ess}\sup \bigl|\Delta_t^i-\Delta_t^j \bigr|<\infty.$

We define
$$
D_{t}(\epsilon):= \e^{\epsilon \Delta}D_{t}+\epsilon A_{t}.
$$
For $\epsilon>0$ the i.i.d.\  sequence $D_{t}(\epsilon)$
satisfies the conditions for the existence of a stochastic growth rate independent
of starting condition.\cite{jC79} That is, if we define the partial products
$$
R_{T}(\epsilon):= D_{T}(\epsilon)\cdot D_{T-1}(\epsilon)\cdot\cdots\cdot D_{1}(\epsilon)
$$
then
$$
a(\epsilon):=\lim_{T\to\infty} T^{-1}\log R_{T}(\epsilon)_{ij}
$$
are well defined deterministic quantities, in the sense that the limit exists
almost surely, is almost-surely constant, and is the same for any 
$0\le i,j\le d-1$.

Of course, $R_{T}(0)$ is not so simple. The off-diagonal terms
are all 0, while on the diagonal, by the Strong Law of Large Numbers,
$$
\lim_{T\to\infty} T^{-1}\log R_{T}(0)_{ii}=\mu_{i}.
$$

\subsection{The effect of the penalty $\Delta$}
We will mostly be concerned with analyzing the case $\Delta\equiv 0$.
For most purposes, $\Delta$ has no effect. 
But this is not always true.

The crucial point is that the effect of $\Delta$ is always nearly linear in $\epsilon$,
while the increase of $a$ near 0 is often superlinear, growing as $\epsilon^{\beta}$.
If the power $\beta$ is strictly less than 1, the rapid increase in $a$ near 0 will be qualitatively
unaffected by a linear term for $\epsilon$ sufficiently small. 
Even when the linear term is negative (as we will generally be assuming it to be),
the growth rate $a$ will still be increasing on a small interval of $\epsilon>0$.

On the other hand, as discussed in section \ref{sec:mainresult} in some cases
we cannot exclude the possibility that the growth rate when $\Delta\equiv 0$
is qualitatively like $\epsilon^{\beta}$ with $\beta\ge 1$. If $\beta>1$ and
$\Delta_{0}<0$ then $a$ will be decreasing near 0; if $\beta=1$ then a
more sensitive analysis would be required.

Since the upper and lower bounds on the appropriate power of $\epsilon$ in 
Theorem \ref{T:diffrate1} are distinct, 
with the lower bound on the growth rate (the upper bound on the power of 
$\epsilon$) being sometimes larger than 1, the current
results will not always permit us to ascertain
whether the growth rate increases or decreases for small
increases in $\epsilon$.

\subsection{Main result} \label{sec:mainresult}
If $d=2$ we have an upper bound that $a(\epsilon)-a(0)$ is smaller than
$\epsilon^{4\tmu^{(1)}/(2\tmu^{(1)}+\tau^{(1)})}$, and lower bound $\epsilon^{4\tmu^{(1)}/\tau^{(1)}}$.
For $d>2$ this becomes slightly more complicated for two reasons: First, the
growth will be dominated by one dimension that has the fastest growth; second,
the increment to growth will be smaller if direct transition between the best two
sites is impossible. For this purpose, for each $1\le j\le d-1$ we define $\kappa_{j}$
to be the smallest length of a cycle in $\M$ that starts and ends at 0, and passes through $j$.
(Thus $\kappa_{j}\ge 2$, and is equal to 2 when $A(0,j)>0$ 
and $A(j,0)>0$ both with positive probability.) Define also
\begin{equation} \label{E:rho}
\rho^{(j)}:= \frac{\tmu^{(j)}}{\tau^{(j)}}\, .
\end{equation}
The calculation of $\rho$ is illustrated in Figure \ref{F:rho}.
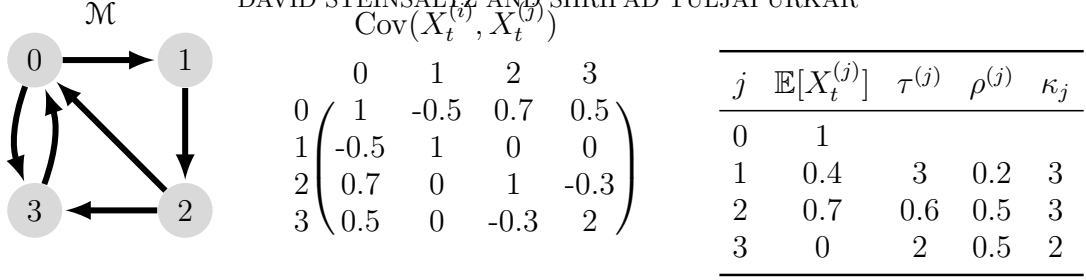
\begin{figure}[ht]
	\begin{center}
		\begin{minipage}[t]{0.15\textwidth}
			\vspace*{-2.7cm}
			$$\M$$
			\begin{tikzpicture}[scale=1, transform shape]
			\tikzstyle{every node} = [circle, fill=gray!30]
			\node (3) at (0, 0) {3};
			\node (2) at +(2,0) {2};
			\node (0) at +(0,2) {0};
			\node (1) at +(2,2) {1}; 
			\foreach \from/\to in {2/0, 2/3, 0/1, 1/2}
			\draw[every loop,
			auto=right,
			line width=.8mm,
			>=latex,
			draw=black,
			fill=black]
			(\from)		edge		 (\to);
			\draw[every loop,
			auto=right,
			line width=.8mm,
			>=latex,
			draw=black,
			fill=black]
			(0)		edge[bend right=20]  (3)
			(3)		edge[bend right=20]  (0);
%
			\end{tikzpicture}
		\end{minipage}
		\begin{minipage}[t]{0.4\textwidth}	
					\vspace*{-2.5cm}
			\begin{center}$\on{Cov}(X_t^{(i)}, X_t^{(j)})$
				$$
				\begin{blockarray}{ccccc}
				& 0 & 1 & 2 & 3 \\
				\begin{block}{c (c c c c)}
				0 & 1 & $-0.5$ & $0.7$ & $0.5$\\
				1 & $-0.5$ & 1 & 0 & 0\\
				2 & $0.7$ & 0 & 1 & $-0.3$\\
				3 & $0.5$ & 0 & $-0.3$ & 2\\
				\end{block}
				\end{blockarray}
				$$
			\end{center}
		\end{minipage}
			\begin{minipage}[t]{0.3\textwidth}
		\begin{center}	\begin{tabular}{c c c c c}
				\toprule
				$j$&$\E[X_t^{(j)}]$&$\tau^{(j)}$&$\rho^{(j)}$& $\kappa_j$\\\midrule
				0 & 1 &  & &\\
				1 & $0.4$ & 3 & $0.2$ & 3 \\
				2 & $0.7$ & $0.6$ & $0.5$ & 3\\
				3 & 0 & $2$ & $0.5$ & 2 \\\bottomrule\\
			\end{tabular}
		\end{center}
	\end{minipage}
	\end{center}
\caption{Calculating $\rho$ on a four-site graph. We
see that both $\kappa\rho$ and $\kappa\rho/(1+2\rho)$ are minimized at site 1, despite the fact that it is
not in the shortest cycle, nor does it have the
smallest mean difference in log growth rate from
the optimal site 0.}
\label{F:rho}
\end{figure}

\begin{Thm} \label{T:diffrate1}
Under the assumptions of section \ref{sec:assumptions}, let $j$ be the site that minimises $\kappa_{j}\rho^{(j)}$, and
$j'$ the site that minimises $\kappa_{j'}\rho^{(j')}/(1+2\rho^{(j')})$.
If $\Delta_{0}=0$ then for any $\constcp>0$ there are positive constants $C,C'$ (depending
on the $\kappa, \rho,\rho_{*},d, \mu_{A},\ttau_{A}$) such that
for all $\epsilon>0$ sufficiently small,
\begin{equation} \label{E:polyepsMoCd3}
\frac{C}{\log \epsilon^{-1}} \epsilon^{2\kappa_{j}\rho^{(j)}} \le
a(\epsilon)-a(0)\le
C' \bigl(\epsilon \log\epsilon^{-1}\bigr)^{2\kappa_{j'}\rho^{(j')}/(1+2\rho^{(j')})} (\log \epsilon^{-1})^{\constcp}.
\end{equation}

Suppose now $\E[\Delta_{0}]< 0$. Then
\begin{itemize}
\item If $2\kappa_{j}\rho^{(j)}<1$ then both
bounds in \eqref{E:polyepsMoCd3} still hold;
\item If 
$
\frac{2\kappa_{j'}\rho^{(j')}}{1+2\rho^{(j')}}<1
   \le  2\kappa_{j}\rho^{(j)}
$ 
then the upper bound in \eqref{E:polyepsMoCd3} holds;
\item If $\frac{2\kappa_{j'}\rho^{(j')}}{1+2\rho^{(j')}}>1$ then
$a$ is differentiable at $0$, with $a'(0)=-\Delta_{0}$.
\end{itemize}
\end{Thm}

In the Appendix we discuss that the requirement for these

Note that $\rho^{(j)}=\infty$ when the distribution of $\xi_{t}^{(j)}$
is not heavy-tailed --- for example, very natural choices such as gamma-distributed diagonal elements --- 
making the lower bound on the left-hand side vacuous,
but it remains an open question whether zero subvariance (see the Appendix for definitions) implies that
the approach of $a(\epsilon)$ to 0 is faster than polynomial in $\epsilon$.

\section{Excursion decompositions}   \label{sec:Excursions}
Since we are assuming the unique maximum average growth rate is at site 0,
the maximum growth for the perturbed process will arise from rare excursions away from 0;
in particular, from those that include the (not necessarily unique) site that minimises $\rho\kappa$ in \eqref{E:rho}.

\nc{\hE}{\hat{\mE}}
\nc{\ee}{\mathbf{e}}
\nc{\he}{\hat{\ee}}
\nc{\Kappa}{K}
Define $\mE$ to be the set --- called the {\em excursions from 0} --- 
of cycles in the migration graph that start and end at 0, with no intervening returns to 0.
For an excursion $\ee$ we write $|\ee|$ for the length of the cycle minus 2 --- that is,
the number of time steps spent away from 0.

For a given excursion $\ee$ we define
\begin{align*}
\Kappa(\ee)&:= \bigl\{ 0\le t\le |\ee|+1 \, : \, \ee_{t}\ne \ee_{t+1} \bigr\} \\
\kappa(\ee)&:=\max\bigl\{ \kappa_{j}\, :\, j\in \ee\bigr\};\\
\rho(\ee)&:=\min\left\{\rho^{(j)}\, :\, j\in \ee\right\}.
\end{align*}
Note that 0 and $T$ are always in $\Kappa(\ee)$, and 
the definition of $\kappa_j$ implies that $\kappa(\ee)\le \#\Kappa(\ee)$.
We will refer to $\kappa(\ee)$ as the {\em diameter} of $\ee$.

We write $\hE_{T}$ for the collection of sequences of excursions that can be fit into
time $\{1,\dots,T\}$. That is, an element $\he\in\hE_{T}$ has an {\em excursion count} $k(\he)$, such that
each $i\in \{1,\dots,k(\he)\}$ there is a pair $(t_{i},\he_{i})$ with $t_{i}\in \{2,\dots,T-1\}$ and 
$\he_{i}\in \mE$ satisfying
\begin{align*}
t_{i}+|\he_{i}|&< t_{i+1},\\
t_{k(\he)}+|\he_{k(\he)}|&\le T.
\end{align*}
We write the total length of an excursion sequence as 
$$
\|\he\|:=\sum_{i=1}^{k(\he)} |\he_{i}|.
$$
We also write $\hE_{T;k,n,m}$ for the subset of $\hE_{T}$ comprising excursion sequences
whose excursion count is $k$, whose total length is $n$, and the sum of whose change-point counts
$\# K(\he_{i})$ is $m$. The {\em null excursion sequence} is the element of $\hE_{T}$ with
$k(\he)=\|\he\|=0$. We illustrate an excursion sequence
in Figure \ref{F:mixed_diameter}.

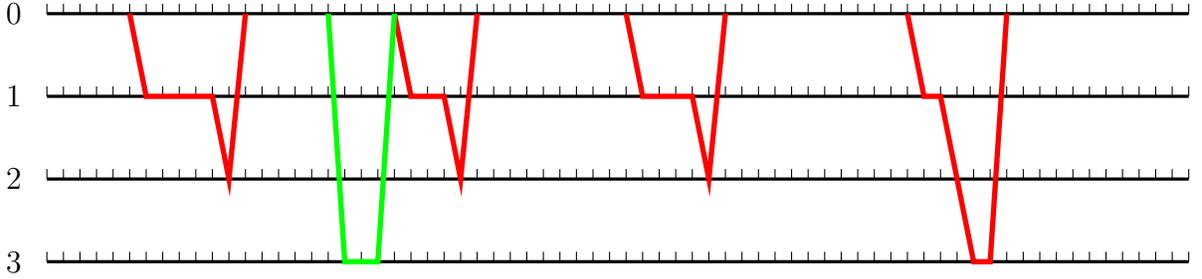
\begin{figure}[ht]
	\begin{tikzpicture}[scale=.22]
	\foreach \y in {0,...,3}{
		\draw[very thick] (0,-5*\y) -- (69,-5*\y);
		\draw (-2,-5*\y) node{\y};
		\foreach \x in {0,...,69}
		\draw (\x , 0-5*\y) -- (\x , 0.6-5*\y);
	}
	\draw[red,line width=2pt] (5,0) -- (6,-5) -- (10,-5) -- (11,-10) -- (12,0);
	\draw[red,line width=2pt] (35,0) -- (36,-5) -- (39,-5) -- (40,-10) -- (41,0);
	\draw[red,line width=2pt] (21,0) -- (22,-5) -- (24,-5) -- (25,-10) -- (26,0);
	\draw[green,line width=2pt] (17,0) -- (18,-15) -- (20,-15) -- (21,0);
	\draw[red,line width=2pt] (52,0) -- (53,-5) -- (54,-5) -- (55,-10) -- (56,-15) -- (57, -15) -- (58,0);
	\end{tikzpicture}
	\caption{An excursion sequence for $T=70$ comprising $k=5$ excursions. This is based on the
		migration graph example from Figure \ref{F:rho}.
		Three excursions (red) have diameter 3, and one
		(green) has diameter 2. Note that one timepoint ($t=21$) is included in two different excursions. The red excursions all have
	$\rho(\he)= 0.2$, and the green excursion has $\rho(\he)=0.5$. The lengths are 6, 2, 3, 5, 5, giving the sequence a total length $n=21$. The
change-point counts are 3, 2, 3, 3, 4, summing to $m=15$.}
		\label{F:mixed_diameter}
	\end{figure}

The $(0,0)$ entry of the product $R_{T}$ will be a sum of terms that are enumerated by elements of
$\hE_{T}$, corresponding to paths through the sites. We define new random variables
as a function of the realizations of $\bX$ and of $\bA$ (the collection of all matrices $A$)
\begin{equation} \label{E:alphas}
\alpha_{t}(i,j):=
\begin{cases}
 \log\epsilon+\log A_{t}(j,i)- X_{t}^{(0)}  &\text{ if } i\ne j,\\
0&\text{ if } i=j.
\end{cases}
\end{equation}
Given an excursion $\ee$ and a starting time $t_{0}\in \{2,\dots,T-|\ee|\}$ we define the random variables
\begin{equation} \label{E:eebracket}
\begin{split}
\ee[t_{0};\bX,\bA]&:=\sum_{t\in \Kappa(\ee)}
 \alpha_{t+t_{0}}(\ee_{t},\ee_{t+1}) 
   +\sum_{t\in \{1,\dots,|\ee|\}\setminus\Kappa(\ee)} \tX_{t+t_{0}}^{(\ee_{t})} \, ,\\
\ee^\Delta[t_{0};\bX,\bA]&:=-\sum_{t\in \Kappa(\ee)} \Delta_{t_0+t}
+ \sum_{t\in \{1,\dots,|\ee|\}\setminus\Kappa(\ee)}  \left( \Delta_t^{\ee_{t}} - \Delta_t^{0} \right) .
\end{split}
\end{equation}
Of course, this sum may be $-\infty$, if it includes a transition at which the corresponding entry of $A$ is 0. But the assumptions imply that it
is finite with nonzero probability if $\ee\in\mE$. Given an excursion sequence $\he=\bigl( (t_{i},\he_{i}) \bigr)_{i=1}^{k}\in \hE_{T}$, we define
\begin{equation} \label{E:hesum}
\he[\bX,\bA]:=\sum_{i=1}^{k} \he_{i}^\Delta[t_{i};\bX,\bA], \quad \text{ and } \he[\bX,\bA]:=\sum_{i=1}^{k} \he^\Delta_{i}[t_{i};\bX,\bA].
\end{equation}

The quantity we are trying to approximate is
\begin{equation} \label{E:aepsilon0}
a(\epsilon)-a(0)=\lim_{T\to\infty} T^{-1}\Bigl(\log R_{T}(\epsilon)_{0,0}-\sum_{i=1}^{T}X_{i}^{(0)}
\Bigr).
\end{equation}

\begin{lemma}\label{L:allterms}
\begin{equation} \label{E:allterms}
\log R_{T}(0,0)=\sum_{t=1}^{T}X_{t}^{(0)}+\epsilon 
\sum_{t=1}^T \Delta_{t}^0
  +\log \Bigl(1+  \sum_{\he\in \hE_{T}\setminus\{\he^{0}\}} \e^{ \he[\bX,\bA] + 
  \epsilon\he^\Delta[\bX,\bA]} \Bigr),
\end{equation}
where $\he^{0}$ is the null excursion sequence.
\end{lemma}

\begin{proof}
We have, by definition,
$$
R_{T}(0,0)=\sum_{(x_{0},\dots,x_{T})} \prod_{t=1}^{T} D_{t}(\epsilon)(x_{t},x_{t-1}),
$$
where the summation is over $(x_{0},\dots,x_{T})\in \{0,\dots,d-1\}^{T+1}$ with
$x_{0}=x_{T}=0$. Note that we may restrict the summation to $(T+1)$-tuples such
that $D_{t}(\epsilon)(x_{t},x_{t-1})>0$, which will only be true when $(x_{t-1},x_{t})$ is an edge of $\M$.
Such sequences of states map one-to-one onto excursion sequences. The product corresponding to excursion sequence 
$\he=\bigl( (t_{i},\he_{i})\bigr)_{i=1}^{k}$ is
\begin{equation} \label{E:prodhe}
\prod_{i=1}^{k+1}\Bigl(\prod_{t=t_{i-1}+1}^{t_{i}} D_{t}(\epsilon)(0,0) \Bigr) 
   \cdot \prod_{i=1}^{k}\Bigl( \prod_{t=1}^{|\he_{i}|} D_{t+t_{i}}(\epsilon) \bigl((\he_{i})_{t}, (\he_{i})_{t+1}\bigr) \Bigr),
\end{equation}
where $t_{0}=0$ and $t_{k+1}=T$.

We have $D_{t}(\epsilon)(0,0)=\e^{X_{t}^{0)}+\epsilon\Delta_{0}}$. Thus, we may write the log of the expression in \eqref{E:prodhe} as
\begin{equation} \label{E:prodhe2}
\sum_{t=1}^{T}  \bigl( X_{t}^{(0)}+\epsilon\Delta_{0} \bigr)  - \sum_{i=1}^{k} \sum_{t=1}^{|\he_{i}|} 
   \log \frac{D_{t+t_{i}}(\epsilon) \bigl((\he_{i})_{t}, (\he_{i})_{t+1}\bigr)}{D^{*}_{t+t_{i}}(0,0)}
\end{equation}
We note that
$$
\log \frac{D_{t+t_{i}}(\epsilon) (j,j)}{D_{t+t_{i}}(\epsilon)(0,0)}=
     \tX_{t+t_{i}}^{(j)}+\epsilon \left(
     \Delta_{t+t_i}^j - \Delta_{t+t_i}^0 \right)
$$
and for $j\ne j'$,
$$
\log \frac{D_{t+t_{i}}(\epsilon) (j,j')}{D_{t+t_{i}}(\epsilon)(0,0)}=
     \log\epsilon A_{t}(j,j')-X_{t}^{(0)}-\epsilon \Delta_{t+t_i}^0 .
$$
Since $\Kappa(\he_{i})$ is precisely the set of $t$ such that $(\he_{i})_{t}\ne (\he_{i})_{t+1}$,
this means that \eqref{E:prodhe2} is precisely the same as $\he_{i}[t_{i};\bX,\bA]$,
which completes the proof.
\end{proof}

Thus
\begin{equation}  \label{E:firstLB}
\log R_{T}(0,0)-\epsilon \sum_{t=1}^T\Delta_t^{0}- \sum_{t=1}^{T} X_{t}^{(0)} \ge  \max_{\he\in\hE_{T}} \he[\bX,\bA] -
 \max_{\he\in\hE_{T}} \he^\Delta[\bX,\bA] ,
\end{equation}
and
\begin{equation}  \label{E:firstUB}
\begin{split}
\log R_{T}(0,0)&-\epsilon \sum_{t=1}^T\Delta_t^{0}- \sum_{t=1}^{T} X_{t}^{(0)}\\
  &\le 3\log T+\max_{1\le k,n,m\le T} \Bigl( \log \#\hE_{T;k,n,m} +\max_{\he\in\hE_{T;k,n,m}} \he[\bX,\bA]  \Bigr)  + \max_{\he\in\hE_{T}} \he^\Delta[\bX,\bA].
\end{split}
\end{equation}
Combining this with \eqref{E:aepsilon0} yields the bounds we will use:
\begin{equation}  \label{E:secondLB}
a(\epsilon)-a(0) \ge  \liminf_{T\to\infty}T^{-1}\max_{\he\in\hE_{T}} \he[\bX,\bA]  - \epsilon \|\Delta\|,
\end{equation}
and
\begin{equation}  \label{E:secondUB}
\begin{split}
a(\epsilon)&-a(0) 
 \le \limsup_{T\to\infty}T^{-1} \max_{1\le k,n,m\le T}\Bigl( \log  \#\hE_{T;k,n,m} +\max_{\he\in\hE_{T;k,n,m}} \he[\bX,\bA ]  \Bigr) 
 +\epsilon \|\Delta\|.
\end{split}
\end{equation}

\section{Derivation of the upper bound}  \label{sec:distinctUB}
We prove the upper bound in \eqref{E:polyepsMoCd3}. We may replace $A_{t}(i,j)$ by
$A_{t}(i,j)\vee 1$ for any $(i,j)\in\M$, since
decreasing $A_t$ can only decrease $a(\epsilon)-a(0)$.
That is, we put a floor under those off-diagonal elements which are
allowable migrations. This avoids the nuisance of having entries
be sometimes 0, and an upper bound that holds under these conditions
will hold {\em a fortiori} under the original conditions. Indeed, we may assume without loss of generality that all $A_t(j,i)=1$ identically for $i\ne j$, since $a(\epsilon)$ --- the stochastic growth rate with the correct values of
$A_t$ --- is no larger than $a(A_* \epsilon; \mathbf{1})$, the stochastic growth rate where all
values of $A_t$ are replaced by 1. This changes our
upper bound only by a constant, which may be absorbed into
the constant of the theorem. Thus, we will proceed under
this assumption.


An element of $\hE_{T;k,n,m}$ may be determined by the following choices:
\begin{enumerate}
\item Choose $k$ points out of $T$ where the excursions begin,
yielding no more than $\binom{T}{k}$ possibilities;
\item Choose $k$ numbers for the lengths of the excursions that add up
to $n$, yielding no more than $\binom{n}{k}$ possibilities;
\item Choose $m-2k$ timepoints within these excursions as times when there
is a change of site, yielding at most $\binom{n}{m-2k}$ possibilities;
\item There are no more than $d^{m}$ ways to choose the sites to which
the excursions move at the $m$ times when there is a change.
\end{enumerate}
A crude bound based on Stirling's Formula is
$$
\log\binom{a}{b}\le b+b\log\frac{a}{b},
$$
which holds for all positive integers $b$ and $0\le a\le b$, 
as long as we adopt the convention $0\cdot \log 0=0\cdot \log \infty=0$.
Then
\begin{equation} \label{E:boundlogFTkn}
\log \#\hE_{T;k,n,m} \le m\log d+(m-2k)\log\frac{n}{m-2k}
   +k\log\frac{n}{k} +T\log\frac{T}{k}.
\end{equation}

\begin{Claim}  \label{C:smallP}
Suppose that $\rho_j$ and $\kappa_j$ are each minimised at
site $j=1$. For any positive $\constcp>0$, and any
$$
z\ge \bigl( \epsilon\log \epsilon^{-1}\bigr)^{2\kappa_{1}\rho_{1}/(1+2\rho_{1})}
   \cdot (\log \epsilon^{-1})^{\constcp},
$$
we have
\begin{equation} \label{E:theclaim}
\limsup_{T\to\infty}  T^{-1} \log\P\bigl\{ \max_{\he\in\hE_{T;k,n,m}}\bigl( \he[\bX,\mathbf{1}]+\log\#\hE_{T;k,n,m}\bigr)\ge zT\bigr\} <0
\end{equation}
for all $\epsilon>0$ sufficiently small.
\end{Claim}

We prove this claim in section \ref{sec:proveClaim}, and proceed here under this assumption.
This means that
$$
\sum_{T=T_{0}}^{\infty} \P\left\{ T^{-1}\max_{1\le k,n,m\le T}\Bigl(\log\#\hE_{T;k,n,m}+\max_{\he\in\hE_{T;k,n,m}} \he[\bX,\bA] \Bigr)
   \ge z \right\}<\infty.
$$
By the Borel--Cantelli Lemma, this implies that with probability 1 this event
occurs only finitely often. It follows that the limsup is smaller than $z$ almost surely, and hence,
by \eqref{E:secondUB}, that
\begin{equation} \label{E:aepsdone}
a(\epsilon)-a(0)\le  \bigl( \epsilon\log \epsilon^{-1}\bigr)^{2\kappa_{1}\rho_{1}/(1+2\rho_{1})}
\cdot (\log \epsilon^{-1})^{\constcp}.
\end{equation}

It remains only to clear away the assumption that that $\kappa_{j}$ and $\rho_{j}$ are both minimized at site 1.
We do this by stratifying the excursions further by their diameter (recall the definition from section \ref{sec:Excursions}).
Define 
$$
\breve{\rho}(\kappa):=\min\bigl\{ \rho_{j}\, : \, \kappa_{j}\le \kappa\bigr\}.
$$
If $\ee$ is an excursion with diameter $\kappa$, then any site $j$ included in $\ee$
has $\kappa_{j}\le \kappa$, hence also $\rho_{j}\ge \breve{\rho}(\kappa)$. Furthermore,
\begin{equation} \label{E:minjkappa}
 \min_{1\le j\le d-1}2\rho_{j}\kappa_{j}/(1+2\rho_{j})=\min_{2\le \kappa\le d-1} 2\rho(\kappa)\kappa/(1+2\rho(\kappa))
\end{equation}

The maximum in \eqref{E:secondUB} may be written as a maximum over $(k_{2},\dots,k_{d-1})$, representing
the number of excursions whose diameter is $2,3,\dots,d-1$, with
the constraint $\sum k_{\kappa}=k$. We write $\hE_{T;k,n,m}^{(\kappa)}$ for the excursion sequences consisting of 
$k$ excursions, all of which have diameter $\kappa$; and $\hE_{T;(k_{\kappa}),n,m}$ for the set of excursion sequences that have
exactly $k_{\kappa}$ excursions with diameter $\kappa$. 
Then $\hE_{T;(k_{\kappa}),n,m}$ naturally includes the direct sum
of $\hE_{T;k_{\kappa},n,m}^{(\kappa)}$. (A sequence of mixed diameters $\he \in \hE_{T;(k_{\kappa}),n,m}$ may be
decomposed into sequences 
$\he_{(\kappa)} \in \hE_{T;k_{\kappa},n,m}^{(\kappa)}$
of excursions with each particular diameter. Referring back to the example in Figure \ref{F:mixed_diameter},
this would entail making one excursion sequence by
dropping out the green excursions, and a separate one
by dropping out the red excursions.) Thus
\begin{align*}
\max_{\he\in\hE_{T;k,n,m}} \he[\bX,\bA]&=\max_{\sum k_{\kappa}=k} \:\max_{\he\in\hE_{T;(k_{\kappa}),n,m}} \he[\bX,\bA]\\
&\le \max_{\sum k_{\kappa}=k} \:\max_{\he_{(\kappa)}\in\hE_{T;k_{\kappa},n,m}^{(\kappa)}} \sum_{\kappa}\he_{(\kappa)}[\bX,\bA]\\
&\le \sum_{\kappa}\max_{1\le  k_{\kappa}\le T} \:\max_{\he_{(\kappa)}\in\hE_{T;k_{\kappa},n,m}^{(j)}} \he_{(\kappa)}[\bX,\bA],
\end{align*}
using the general fact that the maximum of a sum is smaller
than the sum of maxima.
Thus we have
\begin{align*}
\max_{1\le k,n,m\le T} \Bigl(\log\#\hE_{T;k,n,m}+\max_{\he\in\hE_{T;k,n,m}} \he[\bX,\bA] \Bigr)
&\le \sum_{\kappa}\max_{1\le  k_{\kappa},n_{\kappa},m_{\kappa}\le T} \:\max_{\he_{(\kappa)}\in\hE_{T;k_{\kappa},n_{\kappa},m_{\kappa}}^{(\kappa)}} \he_{(\kappa)}[\bX,\bA].
\end{align*}

Because all excursions in $\hE_{T;k_{\kappa},n_{\kappa},m_{\kappa}}^{(\kappa)}$ pass through only sites $j$
with $\rho_{j}\ge \breve{\rho}(\kappa)$, the same argument used for the upper bound in \eqref{E:aepsdone} may be applied to show that
almost surely 
$$
\limsup_{T\to\infty} T^{-1}\max_{1\le  k_{\kappa},n_{\kappa},m_{\kappa}\le T} \:\max_{\he_{(\kappa)}\in\hE_{T;k_{\kappa},n_{\kappa},m_{\kappa}}^{(\kappa)}} \he_{(\kappa)}[\bX,\bA]
   \le c_{\kappa}(\log\epsilon^{-1})^{c'_{\kappa}} \epsilon^{2\kappa\breve{\rho}(\kappa)/(1+2\breve{\rho}(\kappa))}.
$$
It follows that for $c:=(d-2)\cdot \max c_{\kappa}$ and $c':=\max c'_{\kappa}$,
\begin{align*}
\limsup_{T\to\infty} T^{-1}\max_{1\le k,n,m\le T} \Bigl(\log\#\hE_{T;k,n,m}+&\max_{\he\in\hE_{T;k,n,m}} \he[\bX,\bA] \Bigr)\\
   &\le \sum_{\kappa} c_{\kappa}(\log\epsilon^{-1})^{c'_{\kappa}} \epsilon^{2\kappa\breve{\rho}(\kappa)/(1+2\breve{\rho}(\kappa))}\\
   &\le c(\log\epsilon^{-1})^{c'} \epsilon^{\min_{1\le j\le d-1}2\kappa_{j}\rho_{j}/(1+2\rho_{j})}
\end{align*}
by \eqref{E:minjkappa}, which completes the proof.

\section{Derivation of the lower bound}  \label{sec:distinctLB}
We show that the upper bound applies for each $j$;
it will then hold in particular for the $j$ at which $\kappa_{j}\rho_{j}$ attains its minimum.
We may assume without loss of generality that this
optimal site is $j=1$, and
we will write simply $\kappa$, $\tmu$, and $\rho$  for $\kappa_{1}$, $\tmu^{(1)}$,
and $\rho_{1}$.

Let $0=j_{0},j_{1},j_{2},\dots,j_{I}=1,j_{I+1},\dots, j_{\kappa-1},j_{\kappa}=0$,
be a cycle from 0 in $\M$, passing through 1.
We may fix a real number $A_{*}$ and $p>0$ such that 
\begin{equation} \label{E:prodA}
\P\Bigl\{ \sum_{i=0}^{\kappa-1}\bigl(\log X_{t}^{(0)}-\log A_{t+i}(j_{i},j_{i+1}) \bigr)<\kappa A_{*} \, \Bigm| \, \eu{D}\, \Bigr\} \ge p \text{ almost surely},
\end{equation}
where $\eu{D}$ is the sigma-algebra generated by all
the matrices $D_t(0)$.

Assume that $\epsilon\le \e^{-1}$ and $T>\log\epsilon^{-1}$.
Defining $k=\lfloor T/m\rfloor$ and $m=\lfloor \kappa\log\epsilon^{-1}/\tmu\rfloor + \kappa$, we will apply \eqref{E:secondLB} by considering only excursions of length exactly $m-1$,
which proceed exactly through the sequence of sites
$0=j_{0},j_{1},\dots,j_{I-1},1,\dots,1,j_{I+1},\dots,j_{\kappa-1},j_{\kappa}=0$, where site 1 is repeated exactly $m-\kappa+1$ times. 
The basic idea is that the excursion fills a time block of length $m$, proceeding
as quickly as possible from 0 to 1, remaining as long as possible at 1, and then
returning to 0.

We define the standard excursion $\ee_{\circ}:=(j_{1},\dots,j_{I-1},1,\dots,1,
j_{I+1},\dots,j_{\kappa-1})$, with $m-\kappa+1$ repetitions of site 1; and an excursion sequence
$\he_{\circ}$ consisting of those pairs $(\ell m+1,\ee_{\ell})$ for which
\begin{equation} \label{E:fgood}
Y_{\ell}:=\ee_{\circ}\bigl[ \ell m+1; \bX,\bA \bigr] >0.
\end{equation}
That is, $\he_{\circ}$ is put together from identical excursions of form $\ee_{\circ}$
which can start only at times $\ell m+1$. Each one of the $k$
possible excursions is included precisely when its contribution
to the sum would be positive.

We have
$$
\he[\bX,\bA]=\sum_{\ell=0}^{k-1}(Y_{\ell})_{+} \,.
$$
Since the excursion contributions $Y_{\ell}$ are all independent, combining
\eqref{E:secondLB} with the Strong Law of Large Numbers yields
\begin{equation} \label{E:helim}
a(\epsilon)-a(0) \ge \frac{\E[(Y_{\ell})_{+}]}{m} -\epsilon \|\Delta\|.
\end{equation}

We now observe that for any $\ell$
$$
Y_{\ell}
=\sum_{i=0}^{I-1} \alpha_{\ell m+i}(j_{i},j_{i+1})
  +\sum_{i=I}^{\kappa-1} \alpha_{(\ell+1) m-\kappa+i}(j_{i},j_{i+1})
   \,+\, \sum_{t=\ell m+I}^{(\ell+1)m-\kappa+I-1} \tX_{t}^{(1)} .
$$
Note that the $\alpha_{t}$ terms are independent of the $\tX_{t}^{(1)}$ terms.

By \eqref{E:prodA}, for any $y>0$,
\begin{align*}
\P\bigl\{ Y_{\ell}>(m-\kappa +1)y \bigr\} &\ge p\P\Bigl\{ \left(\frac{\tmu}{\rho}(m-\kappa+1)\right)^{1/2}Z > (m-\kappa+1)(y+\tmu)+\kappa\bigl(\log \epsilon^{-1} + A_{*}\bigr) \Bigr\}\\
& \ge p\P\Bigl\{ \left(\frac{\tmu}{\rho}(m-\kappa+1)\right)^{1/2}Z > (m-\kappa+1)(y+ 2\tmu) +\kappa A_{*} \Bigr\}\\
& \ge p\P\left\{ Z > \left( \frac{m\rho}{\tmu} \right)^{1/2} \left(y+ \tmu \left(2 +\delta\right) \right) \right\},
\end{align*}
where
$$
Z:=\left(\frac{\tmu}{\rho}(m-\kappa+1)\right)^{-1/2}\sum_{t=\ell m+I}^{(\ell+1)m-\kappa+I-1} (\tX_{t}^{(1)}+\tmu)
$$
is a standard Gaussian random variable and
$\delta:= A_* /\log \epsilon^{-1}$. 
By Formula 7.1.13 of \cite{AS65} we know that
for all $z\ge 0$,
\begin{equation}  \label{E:normalbound}
	\frac{\e^{-z^2/2}}\ge \frac{\e^{-z^2/2}}{2z+1} \ge \P\left\{Z>z \right\} \ge \frac{\e^{-z^2/2}}{2z+2}.
\end{equation}
Hence for any $z_*>0$ we have
\begin{align*}
	\E\left[ \frac{(Y_\ell)_+}{m-k+1} \right]&
		=\int_0^\infty  \P\bigl\{ Y_{\ell}>(m-\kappa +1)y \bigr\} \dif y\\
	&\ge \int_0^\infty 
	p\P\left\{ Z > \left( \frac{m\rho}{\tmu} \right)^{1/2} \left(y+ \tmu \left(2 +\delta \right) \right) \right\} \dif y\\
	&\ge \frac{pz_*}{(4+2\delta)\tmu +2 z_* + 2}
	\exp\left\{ -\frac{m\rho}{2\tmu} \left( 
	z_* + \tmu(2+\delta)\right)^2 \right\}
\end{align*}
Taking $z_*=1/(4+2\delta)m\rho$ yields
\begin{align*}
	\E\left[ \frac{(Y_\ell)_+}{m-\kappa+1} \right]
		&\ge \frac{p}{ (25\tmu + 10) m\rho +2}
		\exp\left\{ -\frac12(2+\delta)^2 m\rho\tmu -1 \right\}\\
		&\ge \frac{p\e^{-1-3\kappa\rho\tmu}}{\kappa (1+ \rho (25 + 10/\tmu) \log \epsilon^{-1})} \e^{2\kappa\rho(1+\delta/2)^2 \log \epsilon}\\
		&\ge \frac{p\e^{-1-3\kappa\rho\tmu}}{\kappa (1+ \rho (25 + 10/\tmu) \log \epsilon^{-1})} \e^{2\kappa\rho(1+\delta/2)^2 \log \epsilon}\\
		&\ge \frac{p\e^{-1-\kappa\rho(3\tmu+5A_*)}} {\kappa (1+ \rho (25 + 10/\tmu) \log \epsilon^{-1})} \epsilon^{2\kappa\rho}
\end{align*}
for $\epsilon$ sufficiently small that $\delta\le 0.4$
and $m^2\rho^2 \ge \frac18$.

Combining this with \eqref{E:helim}, and assuming
$\epsilon$ small enough that $\tmu/\log\epsilon^{-1}<\frac12$, we have $1-\kappa/m\ge \frac12$ so
$$
	a(\epsilon)- a(0) \ge \frac{C}{\log \epsilon^{-1}}
     \epsilon^{2\kappa\rho}
$$
where
$$
	C= 		\frac{p\e^{-1-\kappa\rho(3\tmu+5A_*)}} {2\kappa (1+ \rho (25 + 5/\tmu)}.
$$

\section{Proof of Claim \ref{C:smallP}} \label{sec:proveClaim}
Since the probability is decreasing in $z$, it will suffice to show the statement is
true for
\begin{equation} \label{E:choosez}
z = \epsilon^{2\kappa_{1}\rho_{1}/(1+2\rho_{1})}
   \cdot (\log \epsilon^{-1})^{\constcp},
\end{equation}
where $\constcp$ is any constant larger than $2\kappa_{1}$.

We define 
$$
\zeta:=\frac{k}{T},\quad\nu:= \frac{n}{k},\quad\beta:=\frac{m}{k}.
$$
That is, $\zeta$ is the rate of excursions per unit time; $\nu$ is the average
length of excursions; and $\beta$ is the average diameter of excursions.
We have the constraints $1/\zeta\ge\nu \ge \beta\ge \kappa_{1} \ge 1$
(since $\kappa_{1}$ is the minimum $\kappa_{j}$, hence the minimum 
number of changes in each excursion).
Then the bound \eqref{E:boundlogFTkn} may be written as
\begin{equation} \label{E:boundlogFTkn2}
\log \#\hE_{T;k,n,m} 
\le \zeta T\left(\beta\log d
   +(\beta-1)\log\nu -(\beta-2)\log(\beta-2)-\log\zeta\right).
\end{equation}

Suppose now we fix some element $\he$ of $\hE_{T;k,n,m}$,
and list all the states of all the excursions in order as $j_{1},\dots,j_{n}$,
we have
$$
\E\bigl[\he[\bX;\indic ]\bigr]\le  m\log\epsilon-\sum_{i=1}^{n} \tmu^{(j_{i})}
$$
and the random variable $Y:=\he[\bX;\indic ] - \E[\he[\bX;\indic ]]$
is Gaussian with variance bounded by $\sum_{i=1}^{n} \ttau^{(j_{i})}$.

For any $x,z>0$, by \eqref{E:normalbound}
\begin{align*}
\P\bigl\{ \bigl( \he[\bX,\indic]&+x\bigr)\ge zT\bigr\} 
\le \P\Bigl\{ Y \ge \sum_{i=1}^{n}\tmu^{(j_{i})} + m\log\epsilon^{-1}
   + zT-x\Bigr\}\\
&\le  \exp\biggl\{-\frac12 \left(\sum_{i=1}^{n} \ttau^{(j_{i})} 
\right)^{-1} \left( \sum_{i=1}^{n}\tmu^{(j_{i})} \,+\,
   m\log\epsilon^{-1}+zT -x\right)^{2} \biggr\}.
\end{align*}
We are assuming that $\rho_{j}$ is minimized
at $j=1$, $\ttau^{(j_{i})}\le \tmu^{(j_{i})}/\rho_{1}$, so that
\begin{align*}
\log\P\bigl\{ \bigl( \he[\bX,\indic]&+x\bigr)\ge zT\bigr\} \\
&\le  -\frac12 \left(\frac{1}{\rho_{1}}\sum_{i=1}^{n} \tmu^{(j_{i})} 
\right)^{-1} \left( 
   \sum_{i=1}^{n}\tmu^{(j_{i})}+m\log\epsilon^{-1}+zT -x\right)^{2}.
\end{align*}

Taking $x=\log\#\hE_{T;k,n,m}$ and substituting \eqref{E:boundlogFTkn2},
we get
\begin{equation} \label{E:Plogbound}
\begin{split}
\log \, \P\bigl\{ & \max_{\he\in\hE_{T;k,n,m}}  \bigl( \he[\bX,\indic]+\log\#\hE_{T;k,n,m}\bigr)
   \ge zT\bigr\} \\
&\le \log\#\hE_{T;k,n,m}+
\max_{\he\in\hE_{T;k,n,m}}\log\P\bigl\{ \bigl( \he[\bX,\indic]+\log\#\hE_{T;k,n,m}\bigr)
   \ge zT\bigr\} \\
&\le T \sup_{S\ge \kappa_1}  \sup_{\beta\ge\kappa_{1}} \sup_{0\le \zeta\le 1}\zeta \biggl[\beta\log dS -\log\zeta
   -(\beta-2)\log(\beta-2)\\
&\hspace*{1cm} - \frac{ \rho_1}{ 2S } \biggl( S +\log\zeta+\frac{z}{\zeta}+(\beta-2)\log(\beta-2)
   -\beta\log dS+\beta\log\epsilon^{-1}  \Bigr)\biggr)^{2} \,\biggr],
\end{split}
\end{equation}
where 
$$
S:=k^{-1}\sum_{i=1}^{n}\tmu^{(j_i)} \ge \nu \tmu_* \ge \beta\tmu_* \ge \kappa_1 \tmu_*
$$
and $\tmu_*:= \min_j \tmu^{(j)}$. Our assumptions ensure that $\tmu_{*}>0$.

We need to show that this supremum is strictly negative.
We write this as $\sup \frac{z}{u} \Theta$, where $u=z/\zeta$ and
\begin{equation} \label{E:Theta}
\begin{split}
 \Theta =& \Theta(S,\beta,u)\\
 : = &-\frac{\rho_1}{2S} \left( S + 
 \beta \log\epsilon^{-1}+\log z +(\beta-2)\log(\beta-2)
 -\beta\log dS -\log u +u\right)^2 \\
 &\hspace*{2cm} - \log z -(\beta-2)\log(\beta-2)
 +\beta\log dS +\log u .
\end{split}
\end{equation}
Here we have taken advantage of the fact that $\log\zeta<0$.
%

We will now show that there are positive constants
$C$, $\Theta_0$, and $\epsilon_0$ (expressible in terms only of
$\constcp,\rho_{1},\kappa_{1},\mu_{*},d$, 
such that for any fixed $\epsilon\in (0,\epsilon_{0})$,
$$
z= \epsilon^{2\kappa_{1}\rho_{1}/(1+2\rho_{1})} (\log \epsilon^{-1})^{\constcp},
$$
and any $S\ge 1$, $\beta\ge \kappa_1$, and $u\ge 0$,
$$
\Theta(S,\beta,u) \le -\Theta_0 - C u.
$$
The result then follows immediately, since then for all $T$,
$$
	T^{-1}\log \, \P\bigl\{  \max_{\he\in\hE_{T;k,n,m}}  \bigl( \he[\bX,\indic]+\log\#\hE_{T;k,n,m}\bigr)
	\ge zT\bigr\} 
	\le zT \sup_{S\ge 1}  \sup_{\beta\ge\kappa_{1}} \sup_{u\ge 0}\frac{1}{u} \Theta(S,\beta,u) \le -Cz .
$$

We consider three different regions for the parameters:
\begin{enumerate}
	\item $S >\log^{2}\epsilon$ and $\beta> \frac{dS}{\log^2 dS} +2$;
	\label{nurange1}
	\item $S >\log^{2}\epsilon$ and $\beta\le \frac{dS}{\log^2 dS} + 2$;
	\label{nurange2}
	\item $\tmu_*\kappa_1\le S \le \log^{2}\epsilon.$
	\label{nurange3}
\end{enumerate}
In range \eqref{nurange1} we have,
$$
  (\beta-2)\log(\beta-2) - \beta\log dS \ge  - 2\log dS 
  	-2\beta \log\log \epsilon^{-1}.
$$
For $\epsilon$ sufficiently small
\begin{align*}
	\Theta &\le  -\frac{\rho_1}{2S} \left( S + (\beta - 2\kappa_1 ) \log\epsilon^{-1} -2\beta \log \log \epsilon^{-1}
	- 2 \log d S +u-\log u\right)^2 \\
	&\hspace*{4cm}+  2\kappa_1 \log\epsilon^{-1}+ 2\log dS  + 2\beta \log\log \epsilon^{-1} + \log u \\
	& \le - \frac{\rho_1}{2} \left( \log^2 \epsilon^{-1} -8\log \log \epsilon^{-1} \right) -\frac{\rho_1}{2} u
	-\beta \left(\rho_1  \log \epsilon^{-1} - 4\log \log \epsilon^{-1}  -2 \kappa_1 \right)
	+(\rho_1 + 1)\log \left(\frac{2\rho_1 + 2}{\mathrm{e}\rho_1} \right) + (\rho_1+ 2) \log d \\
	& \le -\frac{\rho_1}{3} \log^2 \epsilon^{-1}
		-\frac{\rho_1}{2} u
\end{align*}
for $\epsilon$ sufficiently small.

In the range \eqref{nurange2}
\begin{align*}
	\Theta & \le -\frac{\rho_1}{2S} \left( S- \frac{dS}{\log dS}  - \kappa_1 \log \epsilon^{-1} +u -\log u \right)^2
	+ \frac{dS}{\log dS} + \kappa_1 \log \epsilon^{-1} +\log u\\
	&\le -\log^2 \epsilon^{-1}\left( \frac{\rho_1}{2} - \frac{(\rho_1 + 1)d}{\log \log \epsilon^{-1}}\right)
	-\frac{\rho_1}{2} u
	+\kappa_1(\rho_1 + 1) \log \epsilon^{-1} + (\rho_1 + 1)\log \left(\frac{2\rho_1 + 2}{\mathrm{e}\rho_1} \right)\\
	& \le -\frac{\rho_1}{3} \log^2 \epsilon^{-1}
	-\frac{\rho_1}{2} u
\end{align*}
for $\epsilon$ sufficiently small.

In the range \eqref{nurange3} we rewrite $\Theta$ as
$$
\Theta = -\frac{\rho_1}{2S} \left( y+S\right)^2 -y +\beta\log\epsilon^{-1} + u,
$$
where
$$
y:=  \beta \log\epsilon^{-1}+\log z +(\beta-2)\log(\beta-2)
-\beta\log dS -\log u +u .
$$ 
We note that
\begin{align*}
	y&\ge -\beta \log dS + (\beta - \kappa_1) \log \epsilon^{-1} 
	+ \constcp \log \log \epsilon^{-1}+ u-\log u\\
	&\ge  \left(u-\log u \right) + (\constcp- 2\kappa_1) \log \log \epsilon^{-1} -\kappa_1 \log d\tmu_* \\
	&\hspace*{1cm} + (\beta - \kappa_1) \left(\log \epsilon^{-1} -\log\log \epsilon^{-1} - \log d\tmu_* \right)\\
	&\ge 0 
\end{align*}
for $\epsilon$ sufficiently small.
Applying the AM--GM inequality to the first term, we see that
\begin{align*}
  \Theta &\le -(2\rho_1 + 1) y +\beta \log \epsilon^{-1} + u \\
  &\le -2\rho_1 \beta \log\epsilon^{-1} -\rho_1 u -(2\rho_1 + 1) \log z \\
  	&\hspace*{3cm}+ (2\rho_1 + 1) \left[\log \left(1+ \frac{1}{2\rho_1}\right) 
  	-(\beta-2)\log(\beta-2) +\beta\log dS \tmu_* \right]\\
  &\le -\rho_1 u -(\beta- \kappa_1) \left( 2\rho_1 \log \epsilon^{-1}
  -(2\rho_1 +1 ) \log d \log^2 \epsilon^{-1} \right)\\
  &\hspace*{3cm} -(2\rho_1 +1 ) \left[(\constcp - 2\kappa_1 ) \log \log \epsilon^{-1}
  -\mathrm{e}^{-1} - \log \left(1+ \frac{1}{2\rho_1}\right) 
  \right] .
\end{align*}
The last term on the right-hand side is negative 
for $\epsilon$ sufficiently small (and goes
to $-\infty$ as $\epsilon \to 0$);
the same is true of the second term unless $\beta=\kappa_1$,
in which case that term is 0.

\nc{\hbeta}{\beta}

\section{Sub-Gaussian log growth rates} \label{sec:Orlicz}
In our analysis of the case of migration where the optimal site is unique, we have assumed 
that our log growth rates are Gaussian. This is for convenience, simplifying the notation.
In fact, the results depend only on the asymptotic tail behavior. In this section we outline
the modifications that are required for the extension to the sub-Gaussian case.

In \cite{BLM13} a random variable $Z$ is said to be {\em sub-Gaussian} if it
has finite {\em variance factor} $\tau(Z)$, defined as
\begin{equation} \label{E:sGs}
\tau^*(Z):=\inf \left\{ c\ge 0\,:\, \E\left[\e^{\lambda Z}\right] \le 
\e^{c\lambda^{2}/2} \, \forall \lambda \in\R\right\}.
\end{equation}
(The square-root of this 
is called the {\em sub-Gaussian standard} in \cite{BK00}.)
This may be thought of as an upper 
bound on the scale of the tails, and it is this that determines the lower bound on the sensitivity of $a$.
That is, in Theorem \ref{T:diffrate1}  the upper bounds still hold when
the assumption that $\tX_t^{(j)}$ is Gaussian with variance $\tau$ is replaced by sub-Gaussian
with variance factor $\tau^*$.

Similarly, the lower bound on $a(\epsilon)$ only depends on a Gaussian lower bound on the tails
\begin{equation} \label{E:subvar}
\tau_{*}(Z):= \liminf_{z\to\infty}\frac{z^{2}}{-2\log \P\bigl\{ Z>z\bigr\}}
\end{equation}
being nonzero.
That is, in Theorem \ref{T:diffrate1}  the upper and lower bounds still hold when
the assumption Gaussian with variance $\tau$ is replaced by $\tau^*$ and $\tau_*$ respectively.

We point out here that the assumption that $\tX_{t}=\log(\xi_{t}^{(1)}/\xi_{t}^{(0)})$ 
have nonzero $\tau_*$ implies what may be considered exceptionally heavy tails
for the growth rates --- effectively, something like log-normal. 
This is what is required for a nontrivial lower bound in Theorem \ref{T:diffrate1}.
Thus, it seems plausible to infer that the population will obtain no long-term benefit
from sending occasional individuals to a site with lower average growth, unless
the low average growth is compensated by fat positive tails, meaning that there is a small chance
of a very large payoff. (These nearly heavy tails may also be generated if $\xi_{t}^{(0)}$
puts too much probability near 0 --- that is, a population crash.)

The proof of the lower bound can easily be generalized to the sub-Gaussian case, if we replace the specific
calculation of tail probabilities based on the Gaussian distribution 
with a bound based on Cram\'er's Theorem \cite[Theorem 2.2.3]{DZ09}. (The power in the lower bound would need to be increased by an arbitrarily small $\delta$.) The extension of the upper bound of Theorem \ref{T:diffrate1} can be done with the methods
of \cite{dP90} for bounding the tails of maxima
in terms of the {\em Orlicz norm}.
Letting $\Psi(x)=e^{x^{2}}/5$, the Orlicz norm 
$\|Z\|_{\Psi}$ for a centered random variable $Z$ 
is defined to be
\begin{equation} \label{E:orlicz}
\|Z\|_\Psi:=\inf\{C\,:\, \E[\Psi(|Z|/C)]<1\}.
\end{equation}
We present here some elementary results about Orlicz norms and their relationship to the sub-Gaussian variance factors.

\begin{lemma}  \label{L:sGs}
	A sub-Gaussian centered random variable $Z$ satisfies
	\begin{align}
	\label{E:orliczsG}
	\|Z\|_{\Psi}&\le \sqrt{\frac{5 \tau^*(Z)}{2}};\\
	\tau_{*}(Z)&\le \tau^*(Z)<\infty. \label{E:comparevar}
	\end{align}
	
	If $Z$ is Gaussian with mean 0 and variance $\sigma^{2}$ then 
	\begin{equation} \label{E:gausssubgauss}
	\tau_{*}(Z)=\tau^*(Z)=\sigma^{2}.
	\end{equation}
\end{lemma}

\begin{proof}
	The statement \eqref{E:gausssubgauss} is trivial.
	
	If $\tau^*=\tau^*(Z)$ is finite then for any $ \lambda,z,\delta>0$,
	$$
	\P\bigl\{ |Z|>z \bigr\} \le \e^{\frac{(\tau^*+\delta)\lambda^{2}}{2}- \lambda z } .
	$$
	Taking $\lambda=z/(\tau^*+\delta)$, we have $\P\bigl\{ |Z|>z \bigr\} \le \e^{-z^{2}/2(\tau^*+\delta)}$,
	which implies
	\begin{equation} \label{E:taubound}
	\P\bigl\{ |Z|>z \bigr\} \le \e^{-z^{2}/2\tau^*},
	\end{equation}
	since $\delta$ is arbitrary.
	This immediately proves \eqref{E:comparevar}.
	
	Integrating by parts, we have for $C>\sqrt{2\tau}$,
	\begin{align*}
	\E\left[ \e^{Z^{2}/C^{2}} \right] &=1+ \frac{2}{C^{2}}\izf z \e^{z^{2}/C^{2}} 
	\P\bigl\{ |Z|>z \bigr\}\\
	&\le 1+ \frac{2}{C^{2}}\izf z \e^{z^{2}/C^{2}} \e^{-z^{2}/2\tau^*}\\
	&= 1+\frac{2\tau^{*2}}{C^{2}-2\tau^*}.
	\end{align*}
	If $C=\sqrt{5\tau^*/2}$ then this bound is 5, proving \eqref{E:orliczsG}.
\end{proof}

Since it is a norm, the Orlicz norm of an 
arbitrary sum of random variables is no greater than the sum of the Orlicz norms.
For independent sub-Gaussian random variables $X_1,\dots,X_k$
the variance factors are also sub-additive.

\begin{lemma}  \label{L:Orlicztau}
	For any independent centered sub-Gaussian random variables $X_{1},\dots,X_{k}$,
	\begin{equation} \label{E:sGsum}
	\tau^*\Bigl(\sum X_{i}\Bigr) \le \sum \tau^*(X_{i}) ,
	\end{equation}
	and
	\begin{equation} \label{E:sGsum2}
	\P\Bigl\{\Bigl|\sum X_{i}\Bigr|>x\Bigr\} \le \exp\Bigl\{-\left(2\sum \tau^*(X_{i}) \Bigr)^{-1} x^{2} \right\}.
	\end{equation}
	Also
	\begin{equation} \label{E:Orlicztau}
	\| X_1+\cdots+X_k\|_\Psi \le \sqrt{5/2}\left( \sum \tau^*(X_{i}) \right)^{1/2}.
	\end{equation}
	
	If $\max\tau^*(X_{i})\le \tau$ then
	\begin{equation} \label{E:sGsum3}
	\P\Bigl\{\Bigl|\sum X_{i}\Bigr|>x\Bigr\} \le \exp\left\{- \frac{x^{2}}{2k\tau} \right\}.
	\end{equation}
	and
	\begin{equation} \label{E:Orlicztau2}
	\| X_1+\cdots+X_k\|_\Psi \le \sqrt{\frac{5k}{2} \tau }.
	\end{equation}
\end{lemma}

\begin{proof}
	Statement \eqref{E:sGsum} is Lemma 1.7 of \cite{BK00}, and \eqref{E:sGsum2}
	follows by \eqref{E:taubound}. The remainder follows by Lemma \ref{L:sGs}.
\end{proof}

\section{Simulations}  \label{sec:simulation}
We consider a $3\times 3$ example:
$$
M_{t}(\mu,\sigma^{2})=
\begin{pmatrix}
 e^{\sigma Z_{t}^{(0)}}& 0&0\\
0& e^{\sigma Z_{t}^{(1)}-0.1}&0\\
0&0& e^{\sigma Z_{t}^{(2)}-0.2}
\end{pmatrix}, \qquad
A_{t}(C)=
\begin{pmatrix}
0&C&1\\1&0&C\\C&1&0
\end{pmatrix},
$$
with $Z_{t}^{(0)},Z_{t}^{(1)},Z_{t}^{(2)}$ i.i.d.\  standard normal random variables,
and $C$ is a nonnegative constant. If $C=0$ then the migration graph is
a cycle of length 3, so $\kappa_{1}=\kappa_{2}=3$; if $C>0$ then $\kappa_{1}=\kappa_{2}=2$.

We consider three different cases for $(\sigma^{2},C)$: $I: (0.5,1)$, 
$II: (0.5,0)$, and $III: (1,1)$.
We expect to find $\log a(\epsilon)/\log \epsilon^{-1}$
converging to a constant as $\epsilon\downarrow 0$. We have $\tmu^{(1)}=0.1$
in all three cases. For cases I and II we have $\rho^{(1)}=0.1$, so
that the power for case I is between
$$
2\cdot 2 \cdot 0.1=0.4 \qquad \text{and}\qquad \frac{2\cdot 2\cdot 0.1}{1+2\cdot 0.1}
=\frac13,
$$
and for case II is between
$$
2\cdot 3 \cdot 0.1=0.6 \qquad \text{and}\qquad \frac{2\cdot 3\cdot 0.1}{1+2\cdot 0.1}
=0.5.
$$
For case III $\rho^{(1)}$ is decreased to 0.05, so the power is between
$$
2\cdot 2 \cdot 0.05=0.2 \qquad \text{and} \qquad \frac{2\cdot 2\cdot 0.05}{1+2\cdot 0.05}=\frac{2}{11}.
$$
 (Setting $\mu=0$ would put this into the setting of \cite{diapause2018}, with $a(\epsilon)$ behaving like $c/\log\epsilon^{-1}$
for some constant $c$, when $\epsilon$ is small.)

We plot some simulated results in Figures \ref{F:diapausesim1} through  \ref{F:diapausesim3},
plotting the $\log a(\epsilon)$ against $\log\epsilon^{-1}$. In the limit as $\epsilon\to 0$
this should approach a line whose slope is in the range given for the power of $\epsilon$ 
in Theorem \ref{T:diffrate1}. We plot lines with those slopes in each figure, and see that
in the lowest range of $\epsilon$ (we take it down to $\epsilon=10^{-6}$) the slope comes down
close to the upper limit, but is still higher. Of course, this is completely consistent with
the true exponent being at the upper limit, particularly since
we don't know anything yet about how
small $\epsilon$ would need to be before the asymptotic slope becomes apparent.

\begin{figure}[ht]
\begin{center}
\includegraphics[width=10cm]{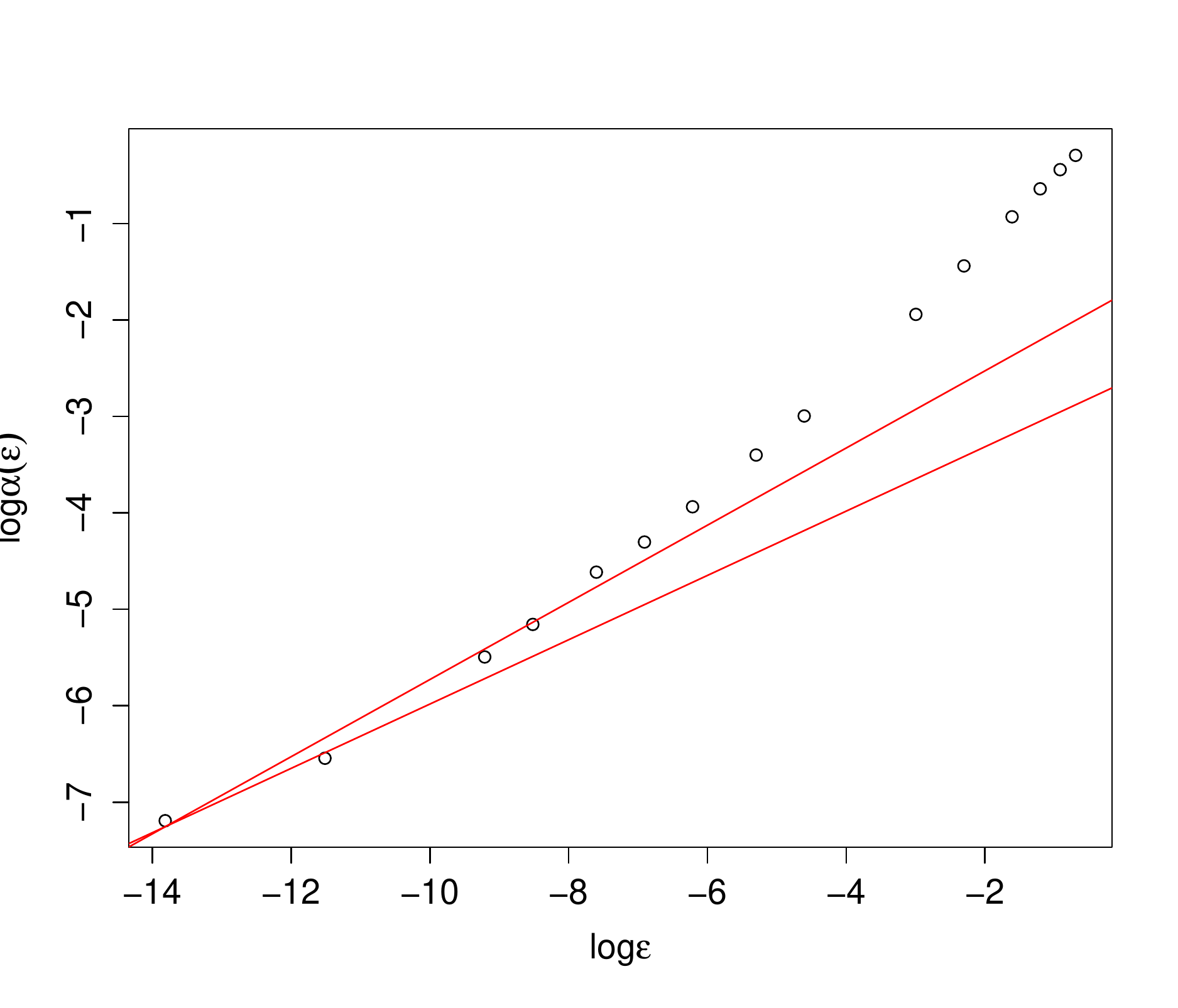}
\caption{Simulated migration example with path length 2. The red lines have slope $0.4$ and $1/3$.}
\label{F:diapausesim1}
\end{center}
\end{figure}

\begin{figure}[ht]
\begin{center}
\includegraphics[width=10cm]{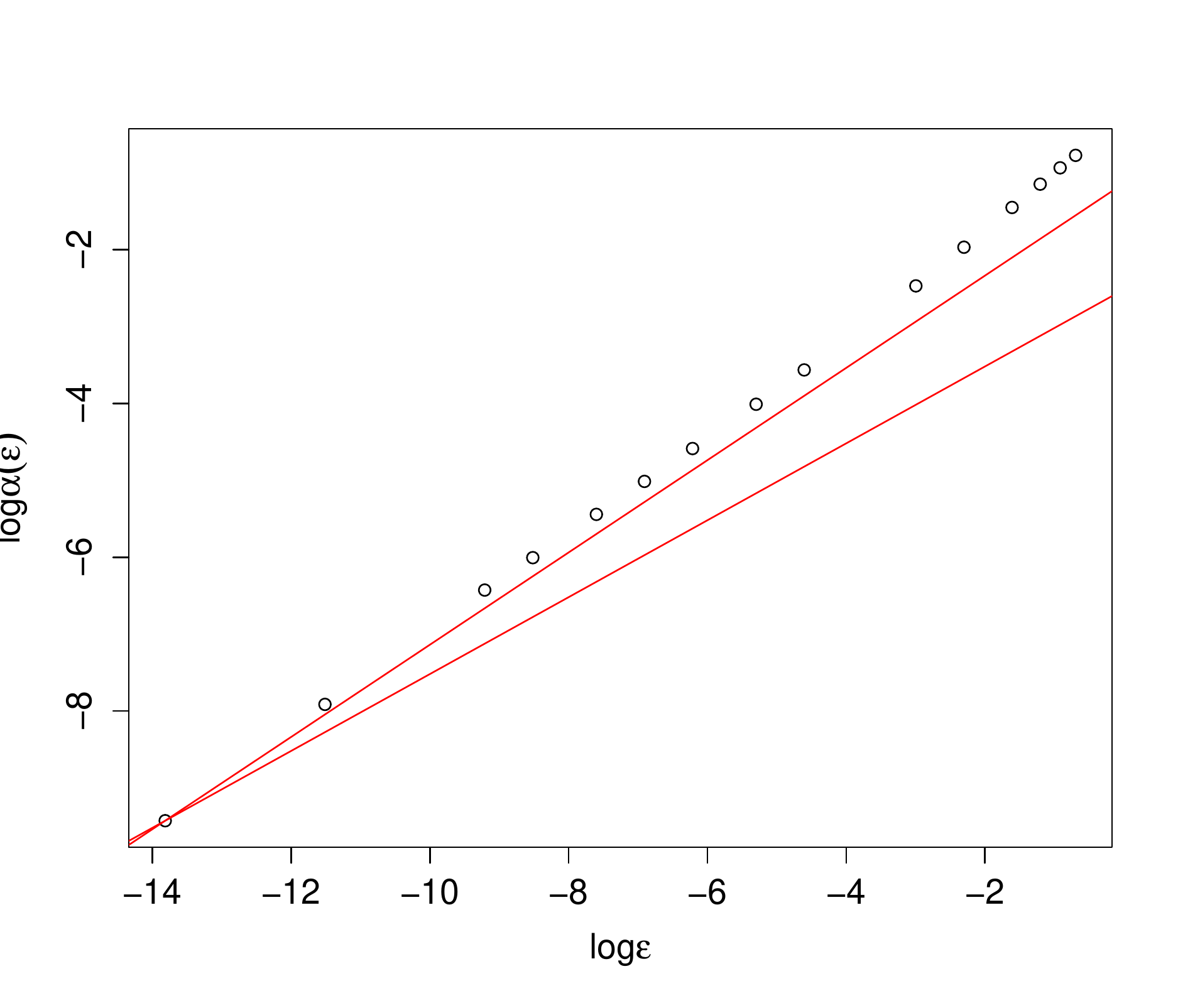}
\caption{Simulated migration example with path length 3. The red lines have slope $0.5$ and $0.6$.}
\label{F:diapausesim2}
\end{center}
\end{figure}

\begin{figure}[ht]
\begin{center}
\includegraphics[width=10cm]{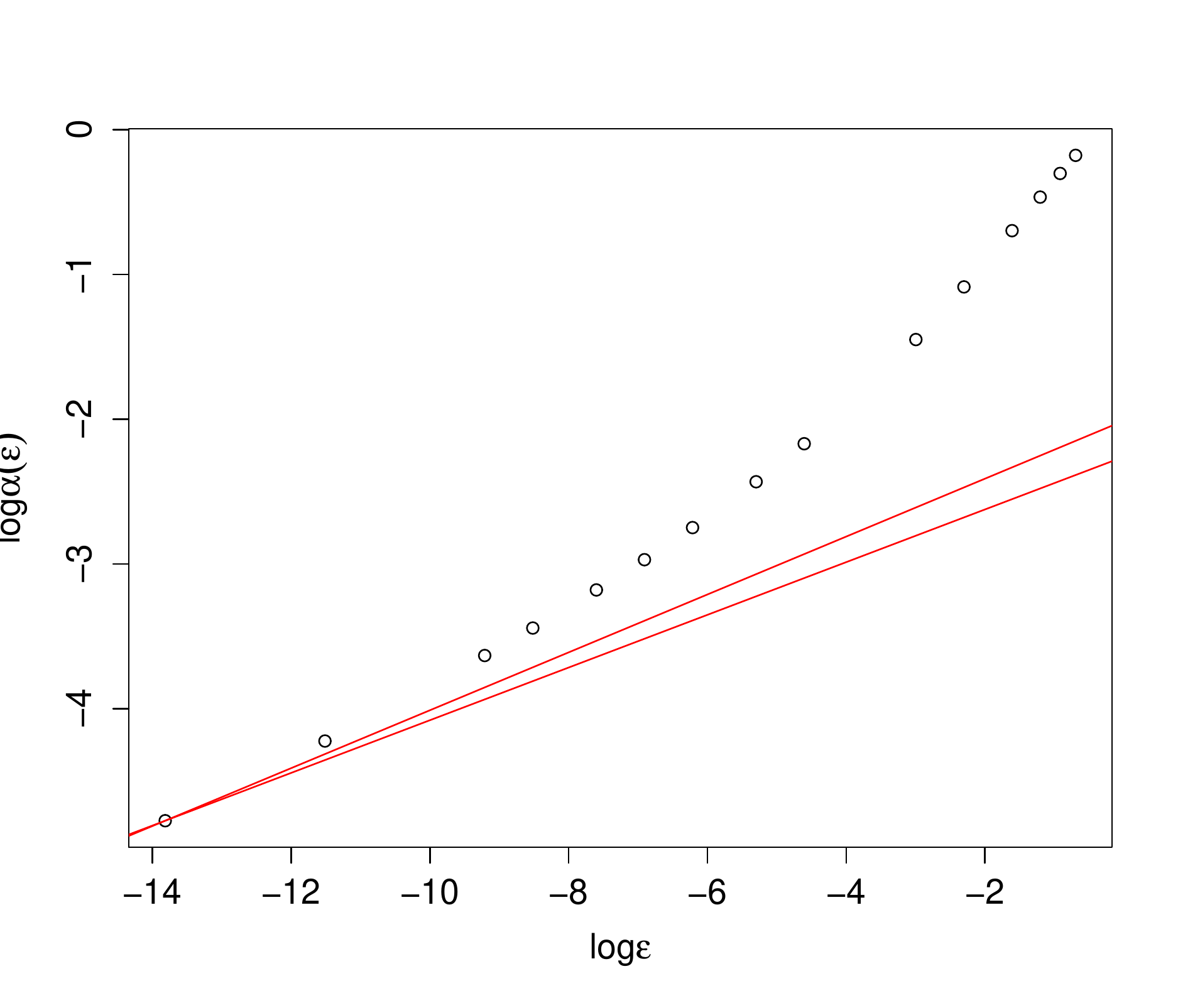}
\caption{Simulated migration example with path length 2, and $\sigma^{2}=1$. The red lines have slope $0.2$ and $2/11$.}
\label{F:diapausesim3}
\end{center}
\end{figure}


\end{document}